%% file: main.tex
\documentclass[10pt,a4paper]{article}
\usepackage[latin1]{inputenc}
\usepackage{amsmath,amsthm}
\usepackage{amsfonts}
\usepackage{amssymb}
\usepackage{graphicx}
\usepackage{geometry}
\usepackage{algorithmic}
\usepackage[]{algorithm2e}

\newcommand{\eps}{\varepsilon}

\newcommand{\Enc}{\mathsf{Enc}}
\newcommand{\Dec}{\mathsf{Dec}}
\newcommand{\polylog}{\mathsf{polylog}}
\newcommand{\poly}{\mathsf{poly}}

\author{Kuan Cheng \thanks{ckkcdh@hotmail.com.\ Department of Computer Science, University of Texas at Austin.\ Supported by a Simons Investigator Award (\#409864, David Zuckerman) and NSF Award CCF-1617713.} \and Xin Li  \thanks{lixints@cs.jhu.edu.\ Department of Computer Science, Johns Hopkins University.\ Supported by NSF Award CCF-1617713 and NSF CAREER Award CCF-1845349.} \and Yu Zheng \thanks{yuzheng@cs.jhu.edu\ Department of Computer Science, Johns Hopkins University.\ Supported by NSF Award CCF-1617713 and NSF CAREER Award CCF-1845349.}}

\title{Locally Decodable Codes with Randomized Encoding}
\newtheorem{theorem}{Theorem}[]
\newtheorem{claim}{Claim}[]
\newtheorem{definition}{Definition}[]
\newtheorem{lemma}{Lemma}[]
\newtheorem{construction}{Construction}[]
\newtheorem*{remark}{Remark}

\begin{document}
	\maketitle
\input{abstract.tex}
\input{introduction.tex}
\subsection{Technique Overview}
\input{techoverview_Hamming.tex}
\input{techoverview_Edit.tex}

\input{prelim.tex}

\input{randEncLDCHamming}	

\input{randEncLDCEdit}

	\section{Discussions and open problems}
	In this work we initiated a study of locally decodable codes with randomized encoding, and we give constructions for both Hamming errors and edit errors, that significantly improve the rate-query tradeoff of known constructions. Our work leaves many interesting open problems. We list some below.
	\begin{description}
	\item[Question 1:] Are there constructions of locally decodable codes with randomized encoding, such that the decoding can succeed without known the randomness of encoding? If so how do such codes perform?
	\item[Question 2:] If the adversary is allowed to know the randomness used by the encoding, is it still possible to achieve much better rate-query tradeoff by using locally decodable codes with randomized encoding?
	\item[Question 3:] Is it possible to construct locally decodable codes with a constant rate, while achieving flexible failure probability, by using randomized encoding?
	\item[Question 4:] Can we further improve the number of queries in our constructions for an oblivious channel?
	\end{description}
Finally, even standard locally decodable codes for edit errors are still are poorly understood. For example, for a constant fraction of error and constant alphabet size, the best approach is still due to \cite{ostrovsky2015locally} which needs at least $\polylog(k)$ queries. Can we prove a lower bound for query complexity in this case, or is there an LDC construction with constant number of queries? 
	
	\bibliographystyle{alpha}
	\bibliography{ldcbib}
\end{document}

%% file: abstract.tex
\begin{abstract}
We initiate a study of locally decodable codes with randomized encoding. Standard locally decodable codes are error correcting codes with a deterministic encoding function and a randomized decoding function, such that any desired message bit can be recovered with good probability by querying only a small number of positions in the corrupted codeword. This allows one to recover any message bit very efficiently in sub-linear or even logarithmic time. Besides this straightforward application, locally decodable codes have also found many other applications such as private information retrieval, secure multiparty computation, and average-case complexity.

However, despite extensive research, the tradeoff between the rate of the code and the number of queries is somewhat disappointing. For example, the best known constructions still need super-polynomially long codeword length even with a logarithmic number of queries, and need a polynomial number of queries to achieve a constant rate. In this paper, we show that by using a randomized encoding, in several models we can achieve significantly better rate-query tradeoff. In addition, our codes work for both the standard Hamming errors, and the more general and harder edit errors.  
\end{abstract}

%% file: introduction.tex
\section{Introduction}
	\par Locally decodable codes (LDCs) are error correcting codes that allow one to decode any specific message symbol by querying only a few symbols from the received codeword, while still tolerating some (say a constant) fraction of errors. This enables the appealing feature of recovering any message symbol very efficiently without reading the entire codeword, and usually this can be done in sublinear time. This feature is especially useful when dealing with large data streams (e.g., databases), and when the recovering of the entire data stream is unnecessary. While locally decodable codes can be defined over an alphabet of any size at least two, in this paper we focus on the case of binary alphabet. Standard locally decodable codes have been studied extensively for Hamming errors, and have found many applications in theory and practice. These include for example private information retrieval, secure multiparty computation, average-case complexity and many more. We refer the reader to \cite{yekhanin2012locally} for an excellent survey on this topic. 
	
	However, in contrast to standard error correcting codes which can achieve excellent rate (the ratio between the length of the message and the length of the codeword), for locally decodable codes the situation is somewhat disappointing. Specifically, if one considers the most natural and interesting case of tolerating a constant fraction of errors, then standard error correcting codes can simultaneously achieve a constant rate. For locally decodable codes, the rate also depends on the number of queries for decoding. Specifically, consider the standard setting where our goal is to recover any message bit with probability $2/3$. With $r$ queries and for message length $k$, it is known that the codeword length needs to be at least $\Omega(k^{(r/r-1)})$ \cite{KT00} and hence if $r$ is a constant,\footnote{For the case of $r=2$ there is a stronger lower bound of $2^{\Omega(k)}$ \cite{KW04}.} the codeword length needs to be at least polynomially long and thus achieving a constant rate is impossible. However, even for the case of $r=O(\log k)$ there is no explicit construction of locally decodable codes with constant rate, and in fact the best known construction in this case fails to achieve codeword length polynomially in $k$. Currently, the best known constructions need  $r=k^{\Omega(1)}$ queries to achieve a constant rate \cite{KSY11}. Below we briefly summarize the parameters of the best known constructions for different regimes of query complexity, while fixing the success probability of decoding any bit to be say $2/3$. For constant query $r=O(1)$, the best known construction is by using matching vector codes \cite{DGY11}, which gives for example a code that uses $r=3 \cdot 2^{t-2}$ queries and can tolerate $\delta=O(1/r)$ fraction of errors. The codeword has length $\mathsf{exp} \mathsf{exp}_t ((\log k)^{1/t}(\log \log k)^{1-1/t})$, which is slightly sub-exponential. In contrast, the lower bound only gives a polynomial codeword length. For query complexity $r=O(\log k)$, the best known construction is also due to matching vector codes, which gives codeword length $\mathsf{exp} \mathsf{exp} (\log \log^{3-o(1)} k)$. Note that this is super-polynomial, while the lower bound already becomes useless. For query complexity $r=\log^t k$ with constant $t> 1$, Reed-Muller codes give codeword length $k^{1+1/(t-1)+o(1)}$. Finally, for query complexity $r=k^{\eps}$ for any constant $\eps>0$, multiplicity codes \cite{KSY11} give codeword length $O(k)$ while tolerating some constant $\delta=\delta(\eps)$ fraction of errors. Thus, there is a huge gap between the upper bounds and the lower bounds. This means that known constructions of locally decodable codes either need a lot of redundancy information, or need to use relatively large number of queries, which is undesirable. Closing this gap is an important open problem, which evidently appears to be quite hard. 
	
	Note that to ensure the property of local decoding, it is necessary to change the decoding from a deterministic algorithm into a randomized algorithm, since otherwise an adversary can just corrupt all the bits the decoding algorithm queries. Hence, the decoding becomes probabilistic and allows some small probability of failure. In this paper, with the goal of improving the rate of locally decodable codes in mind, we initialize the study of a relaxed version of locally decodable codes which is also equipped with randomized encoding. Now the probability of decoding failure is measured over the randomness of both the encoding and decoding. 
	
	There are several questions that we need to clarify in this new model. The first question is: \emph{Does the decoding algorithm know and use the randomness of the encoding?} Although we cannot rule out the possibility of constructions where the decoding algorithm can succeed \emph{without} knowing the randomness of the encoding, in this paper we only consider the most natural setting, where the decoding algorithm indeed knows and uses the randomness of encoding. After all, without the encoding's randomness the decoding algorithm does not even completely know the encoding function. The second question, as we are considering an adversarial situation, is: \emph{Is the adversary allowed to know the randomness of the encoding?}
	
	This turns out to be a very interesting question. Of course, an adversary who knows the randomness of encoding is much stronger, and thus we need a stronger construction of codes as well. In this paper, we provide a partial answer to this question, by showing that under some reasonable assumption of the code, a locally decodable code in our model where the adversary knows  the randomness of encoding is equivalent (up to constant factors) to a standard locally decodable codes. The assumption is that the code has the following homogenous property:
	
	\emph{Property (*) :} For any fixing of the encoding's randomness, any fixed error pattern, and any fixed target message bit, the success probability (over the decoding's randomness) of decoding this bit is the same for all possible messages.
	
	We note that this property is satisfied by all known constructions of standard locally decodable codes. Part of the reason is that all known constructions are linear codes, and the decoding only uses non adaptive queries. We now have the following theorem.
	
	\begin{theorem}
	Suppose there is an LDC with randomized encoding for Hamming errors, with message length $k$, codeword length $n$, that can tolerate $\delta$ fraction of errors and successfully decode any message bit with probability $1-\eps$ using $r$ queries. Then there also exists a (possibly non explicit) standard LDC,  with message length $k/2$, codeword length $n$, that can tolerate $\delta$ fraction of errors and successfully decode any message bit with probability $1-2\eps$ using $r$ queries.
	\end{theorem}
	
	\begin{proof}
	Fix a specific random string used by the encoding, and a target message bit for decoding (say bit $i$). Since now the encoding becomes a deterministic function, and there are only finite number of possible error patterns, there is some error pattern that is the \emph{worst} for decoding, i.e., the one that minimizes the success probability of decoding bit $i$. We say the fixed random string is \emph{good} for bit $i$ if this success probability is at least $1-2\eps$. Note that the overall success probability of decoding bit $i$ is at least $1-\eps$, thus by a Markov argument the probability that a random string is good for bit $i$ is at least $1/2$.
	
	Now again, by an averaging argument, there exists a fixed string that is good for at least $1/2$ fraction of the message bits. Fix this string and the encoding now becomes deterministic. Without loss of generality assume that the string is good for the first half of the message bits. Now fix the rest of the message bits to any arbitrary string (e.g, all $0$), we now have a standard LDC with message length $k/2$, codeword length $n$, that can tolerate $\delta$ fraction of errors and successfully decode any message bit with probability $1-2\eps$ using $r$ queries.
  \end{proof}
  
  The above theorem shows that, in order to get significantly better rate-query tradeoff, we need to either forbid the adversary to know the randomness of encoding, or to construct codes that do not satisfy property (*) (e.g., using adaptive encoding or adaptive decoding). In this paper we study the first setting, which is naturally simpler, and we leave codes in the second setting as an interesting open problem. In general, there are two different models where the adversary is not allowed to know the randomness of encoding:

\paragraph{Shared randomness} In this model, the encoder and the decoder share a private uniform random string. Thus, the adversary does not know the randomness used by the encoder; but he can add arbitrary errors to the codeword, including looking at the codeword first and then adaptively add errors.

\paragraph{Oblivious channel} In this model, the encoder and the decoder do \emph{not} share any randomness. However, the communication channel is oblivious, in the sense that the adversary can add any error pattern \emph{non adaptively}, i.e., without looking at the codeword first.

Here we study both models.\ We now give our formal definitions of locally decodable codes with randomized encoding.
  

\begin{definition}
[LDC with a fixed Failure Probability]
		
		An $(n, k, \delta, q, \eps)$ LDC with randomized encoding consists of a pair of randomized functions $\{\Enc,\Dec\}$, such that:
		
\begin{itemize}
		
	\item $\Enc:\{0,1\}^k\rightarrow \{0,1\}^n$ is the encoding function. For every message $x\in \{0,1\}^k$, $ y = \Enc(x ) \in \{0,1\}^n$ is the corresponding codeword.
			
	\item $\Dec:[k]\times\{0,1\}^* \rightarrow \{0,1\}$ is the decoding function. If the adversary adds at most $\delta n$ errors to the codeword, then for every $i\in [k]$, every $y\in \{0,1\}^*$ which is a corrupted codeword,
		\[ \Pr[\Dec(i,y) = x_i] \geq 1-\eps, \]
	where the probability is taken over the randomness of both $\Enc$ and $\Dec$.
			
	\item $\Dec$ makes at most $q$ queries to $y$.
			
\end{itemize}

\end{definition}

\begin{remark}
Two remarks are in order. First, our definition is quite general in the sense that we don't restrict the type of errors. In particular, in this paper we study both Hamming errors (i.e., bit flips) and edit errors (i.e., insertions, deletions, and substitutions). Second, as mentioned before, the model can either assume shared private randomness or an oblivious adversarial channel.
\end{remark}

The above definition is for a fixed failure probability $\eps$. However, sometimes it is desirable to achieve a smaller failure probability. For standard locally decodable codes, this can usually be done by repeating the decoding algorithm independently for several times, and then taking a majority vote. This decreases the failure probability at the price of increasing the query complexity. In contrast, in our new model, this approach is not always feasible. For example, in the extreme case one could have a situation where for some randomness used by the encoding, the decoding succeeds with probability 1; while for other randomness used by the encoding, the decoding succeeds with probability 0. In this case repeating the decoding algorithm won't change the failure probability. To rule out this situation, we also define a locally decodable codes with \emph{flexible} failure probability.

\begin{definition}
[LDC  with Flexible Failure Probabilities]
		
		An $ (n, k, \delta)$ LDC with randomized encoding and query complexity function $ q: \mathbb{N} \times [0,1] \rightarrow \mathbb{N} $, consists of a pair of randomized algorithms $\{\Enc,\Dec\}$, such that:
		
\begin{itemize}
		
	\item $\Enc:\{0,1\}^k\rightarrow \{0,1\}^n$ is the encoding function. For every message $x\in \{0,1\}^k$, $ y = \Enc(x ) \in \{0,1\}^n$ is the corresponding codeword.
			
	\item $\Dec:[k]\times\{0,1\}^* \rightarrow \{0,1\}$ is the decoding function. If the adversary adds at most $\delta n$ errors to the codeword, then for every $i\in [k]$, every $y\in \{0,1\}^*$ which is a corrupted codeword, and every $ \eps \in [0,1]$, 
		\[ \Pr[\Dec(i, y) = x_i] \geq 1-\eps, \]
	while $\Dec$ makes at most $q = q(n, \eps) $ queries to $y$.
	The probability is taken over the randomness of both $\Enc$ and $\Dec$.

\end{itemize}
				
\end{definition}

Again, this definition can apply to both Hamming errors and edit errors, and both the model of shared randomness and the model of an oblivious adversarial channel.

To the best of our knowledge, locally decodable codes with randomized encoding have not been studied before in the literature. In this paper we provide several constructions with rate-query tradeoff better than standard locally decodable codes. We have the following theorems. The first one deals with a fixed decoding failure probability.

\begin{theorem}
\label{thm:intro1}
There exists a constant $0<\delta<1$ such that for every $k \in \mathbb{N}$ there is an efficient construction of $(n, k , \delta, q,  \eps)$ LDC with randomized encoding,  where $n=O(k)$. For any $\eps \in [0,1]$, the query complexity is as follows.
\begin{itemize}
\item for Hamming errors with shared randomness, $q = O(\log \frac{1}{\eps})$;

\item for Hamming errors with an oblivious channel, $q = O(\log k\log \frac{1}{\eps})$;

\item for Edit errors, both with shared randomness and with an oblivious channel,  $q = \polylog k\log \frac{1}{\eps}$.

\end{itemize} 

\end{theorem}
We can compare our theorem to standard locally decodable codes. Achieving $q = O(\log \frac{1}{\eps})$ is equivalent to achieving constant query complexity for success probability $2/3$. For standard LDCs we know in this case it is impossible to achieve a constant rate, and the best known construction still has sub-exponential codeword length. Similarly, achieving $q = O(\log k \log \frac{1}{\eps})$ is equivalent to achieving query complexity $O(\log k)$ for success probability $2/3$, and for standard LDCs the best known construction still has super-polynomial codeword length. In contrast, we can achieve a constant rate.

The next theorem deals with a flexible decoding failure probability. We note that one way to achieve this is to repeat the encoding several times independently, and send all the obtained codewords together. The decoding will then decode from each codeword and take a majority vote. However, this approach can decrease the rate of the code dramatically. For example, if one wishes to reduce the failure probability from a constant to $2^{-\Omega(k)}$, then one needs to repeat the encoding for $\Omega(k)$ times and the rate of the code decreases by a factor of $1/k$. In this work we use a different construction that can achieve a much better rate.

\begin{theorem}
\label{thm:intro2}
There exists a constant $0<\delta<1$ such that for every $k \in \mathbb{N}$ there is an efficient construction of $(n, k , \delta)$ LDC with randomized encoding and flexible failure probability,  where $n =  O(k \log k)$. The query complexity is as follows.
\begin{itemize}
\item for Hamming errors with shared randomness, $q = O(\log k\log \frac{1}{\eps})$;

\item for Hamming errors with an oblivious channel,  $q = O(\log^2 k\log \frac{1}{\eps})$;

\item for Edit errors, both with shared randomness and with an oblivious channel,  $q = \polylog k\log \frac{1}{\eps}$.

\end{itemize} 

\end{theorem}

We can also compare this theorem to standard locally decodable codes. Again, achieving $q = O(\log k \log \frac{1}{\eps})$ is equivalent to achieving query complexity $O(\log k)$ for success probability $2/3$, and for standard LDCs the best known construction still has super-polynomial codeword length. Similarly, achieving $q = O(\log^2 k \log \frac{1}{\eps})$ is equivalent to achieving query complexity $O(\log^2 k)$ for success probability $2/3$, and for standard LDCs the best known construction using Reed-Muller codes needs codeword length $k^{2+o(1)}$. In contrast, we can achieve codeword length $O(k \log k)$.

\subsection{Related works.}	
As mentioned earlier, standard locally decodable codes for Hamming errors have been studied extensively, with many beautiful constructions such as \cite{Yekhanin08}, \cite{Efremenko09}, \cite{DGY11}, and we refer the reader to \cite{yekhanin2012locally} for an excellent survey.

For edit errors, the situation is quite different. Even for standard error correcting codes (not locally decodable), the progress for edit errors is much slower than that for Hamming errors. For example, the first asymptotically good code for edit errors is not known until the work of Schulman and Zuckerman in 1999  \cite{schulman1999asymptotically}. In recent years, there has been a line of research  \cite{7835185}, \cite{guruswami2016efficiently}, \cite{BukhV16}, \cite{Belazzougui2015EfficientDS}, \cite{haeupler2017synsimucode}, \cite{CJLW18}, \cite{haeupler2018optimal}, \cite{CJLW19} that makes constant progress and finally culminates in \cite{CJLW18}, \cite{haeupler2018optimal} which give codes that can correct $t$ edit errors with redundancy (the extra information needed for the codeword) size $O(t \log^2(k/t))$, where $k$ is the message length. This is within a $\log(k/t)$ factor to optimal. In \cite{CJLW18}, the authors also give a code that can correct $t$ edit errors with redundancy $O(t \log k)$, which is optimal up to constants for $t \leq k^{1-\alpha}$, any constant $\alpha>0$. In another line of work  \cite{haeupler2017synchronization}, \cite{HS17c}, \cite{CHLSW18}, near optimal codes for edit errors over a larger alphabet are constructed, using the constructions of \emph{synchronization strings}.
	
In terms of locally decodable codes, no known constructions for Hamming errors can be used directly as a code against edit error. This is due to the shift caused by insertions and deletions, which causes the loss of information about the positions of the bits queried. The only known previous work of locally decodable codes for edit errors is the work by Ostrovsky and Paskin-Cherniavsky \cite{ostrovsky2015locally}. There the authors provided a ``compiler'' to transform a standard locally decodable code for Hamming error to a locally decodable code for edit error. Their construction first takes a standard LDC as a black box, then concatenate it with an asymptoticaly good code for edit errors. They showed that the fraction of error tolerated and the rate of the resulted code both only decrease by a constant factor. Meanwhile, the query complexity increases by factor of $polylog(n)$, where $n$ is the length of the codeword. 

Standard error correcting codes with randomized encoding (but deterministic decoding) have been studied by Guruswami and Smith \cite{guruswami2016optimal}. They constructed such codes for an oblivious adversarial channel that can add $\delta$ fraction of errors, where the rate of the codes approaches the Shannon capacity $1-H(\delta)$. Their encoding and decoding run in polynomial time and the decoding error is exponentially small.

%% file: techoverview_Hamming.tex
\subsubsection{Hamming Error}
 We start from constructions for Hamming errors and let the fraction of errors be some constant $\delta$. To take advantage of a randomized encoding, our approach is to let the encoder perform a \emph{random permutation} of the codeword. Assuming the adversary does not know the randomness of the encoding, this effectively reduces adversarial errors into random errors, in the following sense: if we look at any subset of coordinates of the codeword before the random permutation, then the expected fraction of errors in these coordinates after the random permutation is also $\delta$. A stronger statement also follows from concentration bounds that with high probability this fraction is not far from $\delta$, and in particular is also a constant. This immediately suggests the following encoding strategy: first partition the message into blocks, then encode each block with a standard asymptotically good code, and finally use a random permutation to permute all the bits in all resulted codewords. Now to decode any target bit, one just needs to query the bits in the corresponding codeword for the block that contains this bit. As the error fraction is only a constant, a  concentration bound for \emph{random permutations} shows that in order to achieve success probability $1-\eps$, one needs block length $O(\log (1/\eps))$ and this is also the number of queries needed. To ensure that the adversary does not learn any information about the random permutation by looking at the codeword, we also use the shared randomness to add a mask to the codeword (i.e., we compute the XOR of the actual codeword with a random string). This gives our codes for a fixed failure probability, in the model of shared randomness. 

To modify our construction to the model of an oblivious channel, note that here we don't need a random mask anymore, but the encoder has to tell the decoder the random permutation used. However the description of the random permutation itself can be quite long, and this defeats our purpose of local decoding. Instead, we use a \emph{pseudorandom permutation}, namely an almost $\kappa = \Theta(\log \frac{1}{\eps})$ wise independent permutation with error $\eps/3$. Such a permutation can be generated by a short random seed with length $r=O(\kappa \log n + \log(1/\eps )) $, and thus is good enough for our application. The encoder will first run the encoding algorithm described previously to get a string $y$ with length $n/2$, and then concatenate $y$ with an encoded version $z \in \{0,1\}^{n/2}$ of the random seed.  The decoder will first recover the seed and then perform local decoding as before. To ensure the seed itself can be recovered by using only local queries, we encode the seed by using a concatenation code where the outer code is a  $(n_1, k_1, d_1)$ Reed-Solomon code with alphabet size $\mathsf{poly}(n)$, and the inner code is an  $(n_2, k_2, d_2)$ asymptotically good code for $O(\log n)$ bits, where $  n_1 n_2 = n/2, k_1 k_2 = r$, $d_1 = n_1 - k_1 + 1$, $n_2 = O(\log n)$. That is, each symbol of the outer code is encoded into another block of $O(\log n)$ bits. Now to recover the seed, the decoder first randomly chooses  $8k_1$  blocks of the concatenated codes and decodes the symbols of the outer code. These decoded symbols will form a new (shorter) Reed-Solomon code, because they are the evaluations of a degree $k_1$ polynomial on $8k_1$ elements in $\mathbb{F}_2^{k_2}$.
Furthermore, this code is still enough to recover the seed, because the seed is short and with high probability there are only a small constant fraction of errors in the decoded symbols. 
So the decoder can perform a decoding of the new Reed-Solomon code to recover the random seed with high probability. Note that the number of queries of this decoding is $8k_1 n_2 = O(r)$.

We now turn to our constructions for flexible failure probability. Here we want to achieve failure probability from a constant to say $2^{-k}$. For each fixed failure probability we can use the previous construction, and this means the block size changes from a constant to $O(k)$. Instead of going through all of these sizes, we can just choose $O(\log k)$ sizes and ensure that for any desired failure probability $\eps$, there is a block size that is at most twice as large as the size we need. Specifically, the block sizes are $ 2^{2^i}, i = 1,2, \ldots, \log k $. In this way, we have $O(\log k)$ different codewords, and we combine them together to get the final codeword. Now for any failure probability $\eps$, the decoder can look for the corresponding codeword (i.e., the one where the block size is the smallest size larger than $ 1/\eps $)  and perform local decoding. However, we cannot simply concatenate these codewords together since otherwise $\delta$ fraction of errors can completely ruin some codeword. Instead, we put them up row by row into a matrix of $O(\log k)$ rows, and we encode each column of $O(\log k)$ bits with another asymptotically good binary code. Finally we concatenate all the resulted codewords together into our final codeword of length $O(k \log k)$. Note that now to recover a bit for some codeword, we need to query a whole block of $O(\log k)$ bits, thus the query complexity increases by a $\log k$ factor while the rate decreases by a $\log k$ factor. 

%% file: techoverview_Edit.tex
\subsubsection{Edit Error}

Our constructions for edit errors follow the same general strategy, but we need several modifications to deal with the loss of index information caused by insertions and deletions. Our construction for edit errors can achieve the same rate as those for Hamming case, but the query complexity for both models increases by a factor of $\polylog k$. We now give more details. We start with the construction of an $(n = O(k), k , \delta = \Omega(1), q = \polylog k \log(1/\eps), \eps)$ LDC for any $\eps\in (0,1)$, in the model with shared randomness. As in the case of Hamming errors, the shared randomness is used in two places: a random permutation $\pi$ and some random masks to hide information. The construction has two layers.

For the first layer, view the message  $x\in \{0,1\}^k$ as a sequence over the alphabet $\{0,1\}^{\log k}$ and divide it into $k/(k_0\log k)$ small blocks each containing $k_0$ symbols from $\{0,1\}^{\log k}$.  Then, we encode each block with $\Enc_0$, which is an asymptotically good $(n_0,k_0,d_0)$ code for Hamming errors over the alphabet $\{0,1\}^{\log k}$. Concatenating these $k/(k_0\log k)$ codewords gives us a string of length $N = \frac{n_0k}{k_0 \log k}$ over the alphabet $\{0,1\}^{\log k} $. We then permute these $N$ symbols using $\pi$ to get $y'=B_1\circ B_2\circ \cdots\circ B_{N}$ with $B_i \in \{0,1\}^{\log k}$. Since $n_0/k_0$ is a constant, we have $N<k$ for large enough $k$.

We are now ready to do the second layer of encoding. In the following, for each $i\in [N]$, $b_i\in \{0,1\}^{\log k}$ is the binary representation of $i$ and $r_i \in \{0,1\}^{\log k}$ is a random mask shared between the encoder and decoder. $\mathbb{C}_0: \{0,1\}^{2\log k}\rightarrow \{0,1\}^{10\log k}$ is an asymptotically good code for edit errors and the $\oplus$ notation means bit-wise XOR. For  each $i\in [N]$, we compute $B'_i = \mathbb{C}_0(b_i\circ (B_i\oplus r_i))\in \{0,1\}^{10\log k}$. We output $y = B'_1\circ  B'_2\circ \cdots \circ B'_{N}\in (\{0,1\}^{10\log k })^{N}$, which is of length $n = 10\frac{n_0}{k_0}k = O(k)$. Note that the use of the random masks $r_i$'s is to hide the actual codeword, so that the adversary cannot learn any information about the permutation. 

To decode a given message bit, we need to find the corresponding block. The natural idea to do this is by binary search. However, this may fail with high probability due to the constant fraction of edit errors. To solve this issue, we use the techniques developed by \cite{ostrovsky2015locally}, which use a similar second layer encoding and give a searching algorithm with the following property: Even with $\delta$-fraction of edit errors, at least $1-O(\delta)$ fraction of the blocks can be recovered correctly with probability at least $1-\mathsf{neg}(k)$. The algorithm makes a total of $\polylog k$ queries to the codeword for each search. 

We now describe the decoding. Assume the bit we want to decode lies in the $i$-th block of $x$. Let $C_i\in (\{0,1\}^{\log k})^{n_0}$ be the codeword we get from encoding the $i$-th block using $\Enc_0$. With the information of $\pi$, we can find out $n_0$ indices $i_1$ to $i_{n_0}$ such that $C_i$ is equal to $B_{\pi^{-1}(i_1)} \circ B_{\pi^{-1}(i_2)} \circ \cdots \circ B_{\pi^{-1}(i_{n_0})}$. The decoding algorithm calls the searching algorithm from \cite{ostrovsky2015locally} to find all blocks $B'_{i_1} $ to $B'_{i_{n_0}}$ in the received codeword. We say a block is unrecoverable if the searching algorithm failed to find it correctly. By the same concentration bound used in the Hamming case and the result from \cite{ostrovsky2015locally}, the fraction of unrecoverable blocks is bounded by a small constant with high probability. Thus, we can decode $C_i$ correctly with the desired success probability. In this process, each search takes $\polylog k$ queries and $n_0 = O(\log 1/\eps)$ of searches are performed. The total number of queries made is thus $\polylog k \log 1/\eps$.

For the model of an oblivious channel, again we use a pseudorandom permutation $\pi$ that can be generated by $O(\log n\log 1/\eps)$ random bits. We use the same binary code as we used in the Hamming case to encode it and then view it as a string over the alphabet $\{0,1\}^{\log k}$. It is then concatenated with the code described previously before the second layer of encoding. After that, the same second layer of encoding is applied. The random masks used in the previous construction are no longer needed since the adversary can not see the codeword.

The construction for a flexible failure probability is also similar to the Hamming case. We write the codes before the second layer of encoding as a matrix $M$. The only difference is that, each element in the matrix $M$ is now a symbol in $\{0,1\}^{\log k}$. We then encode the $j$-th column with an error correcting code over the alphabet $\{0,1\}^{\log k}$ to get a codeword $z_j$, and concatenate them to get $z$. After that, we do the second layer of encoding on $z$. 

%% file: prelim.tex
\section{Preliminaries}
	
	
\subsection{Edit distance}
	
\par Edit distance, or Levenstein distance, is a popular metric for quantifying how similar two strings are. It is defined using three types of string operations: insertion, deletion, and substitution. An insertion operation inserts a character into the string at a certain position. A deletion operation deletes a particular character from the string. And a substitution operation substitutes a character with a different one. 
\begin{definition}
	The edit distance between two strings $x,y\in \Sigma^*$ , denoted by $d_E(x,y)$, is the minimal number of edit operations (insertion, deletion, and substitution) needed to transform one into another. Assume $x$ and $y$ have same length $n$, the normalized edit distance of $x$ and $y$ is $d_E(x,y)$ devided by $n$. We denote the normalized edit distance by $\Delta_E(x,y)$
\end{definition}

\subsection{Pseudorandom permutation}

\par In our construction against oblivious channel, one key idea is to integrate a short description of a pseudorandom permutation into the codeword. We will utilize the construction of pseudorandom permutations from \cite{kaplan2009derandomized}. Here, we follow the definition from their work and introduce their results.
\begin{definition}
	[Statistical Distance]
	Let $D_1$, $D_2$ be distributions over a finite set $\Omega$. The variation distance between $D_1$ and $D_2$ is 
\[\|D_1 -D_2\| = \frac{1}{2}\sum_{\omega\in \Omega}|D_1(\omega)-D_2(\omega)|\]
We say that $D_1$ and $D_2$ are $\delta$-close if $\|D_1-D_2\|\leq\delta$.		
\end{definition}
\begin{definition}
	[$k$-wise $\delta$-dependent permutation]
	Let $\mathcal{F}$ be a family of permutations on $n$ elements (allow repetition). Let $\delta>0$, we say the family $\mathcal{F}$ is $k$-wise $\delta$-dependent if for every $k$-tuple of distinct elements $\{x_1,x_2,\cdots,x_k\}\in [n]$, for $f\in \mathcal{F}$ chosen uniformly, the distribution $\{f(x_1),f(x_2),\cdots,f(x_k)\}$ is $\delta$-close to uniform distribution
\end{definition}
In practice, we want to construction explicit families of permutations. Two related parameters are:
\begin{definition}
	[Description length]
	The description length of a permutation family $\mathcal{F}$ is the number of random bits, used by the algorithm for sampling permutations uniformly at random from $\mathcal{F}$. 
\end{definition}
\begin{definition}
	[Time complexity]
	The time complexity of a permutation family $\mathcal{F}$ is the running time of the algorithm for evaluating permutation from $\mathcal{F}$
			
\end{definition}
It is known that we can construct families of $k$-wise almost independent permutations with short description length (optimal up to a constant factor).
\begin{theorem}
	[Theorem 5.9 of \cite{kaplan2009derandomized}]
	\label{thm:almosttwiserandperm}
	Let $P_n$ denote the set of all permutation over $\{0,1\}^n$. There exists a $k$-wise $\delta$-dependent family of permutation $\mathcal{F}\subset P_n$. $\mathcal{F}$ has description length $O(k n + \log \frac{1}{\delta})$ and time complexity $\poly( n ,k,\log \frac{1}{\delta})$.
\end{theorem}

\subsection{Concentration bound}

In our proof, we use the following concentration bound from \cite{cheng2017near}

\begin{lemma}[\cite{cheng2017near}]
\label{lem:rponeset}
Let $\pi:[n] \rightarrow [n]$ be a random permutation. For any set $S, W \subseteq [n]$, let  $u = \frac{|W|}{n}|S|$. Then the following holds. 
\begin{itemize}
\item for any constant $\delta \in (0, 1)$, 
$$   \Pr[|\pi(S) \cap W| \leq (1- \delta) \mu  ] \leq e^{-\delta^2 \mu/2 },   $$
$$   \Pr[|\pi(S) \cap W| \geq (1+ \delta) \mu  ] \leq e^{-\delta^2 \mu/3 }.   $$

\item  for any $d \geq 6\mu$, $\Pr[|\pi(S) \cap W|  \geq d]\leq 2^{-d}$.

\end{itemize}

\end{lemma}

%% file: randEncLDCHamming.tex
\section{Random Encoding LDC for Hamming Errors}

\subsection{With a  Certain Failure Probability}

We first show the construction of a Random Encoding LDC with an arbitrarily fixed failure probability for the case of shared randomness.

\begin{theorem}
\label{thm:LDCHamming1sr}
There is an efficient construction of $(n, k =\Omega(n) , \delta = \Theta(1), q = O(\log n\log \frac{1}{\eps}),  \eps)$ random LDC for the case of shared randomness, with arbitrarily    $\eps\in [0,1]$. 

\end{theorem}

\begin{construction}
\label{construct:randEncLDCHammingsr}
We construct an $(n, k = \Omega(n), \delta = \Theta(1), q= O(\log n \log \frac{1}{\eps}), \eps)$-locally decodable   binary code for arbitrary $\eps \in [0,1]$, for the case of shared randomness.
 
Let $\delta_0, \gamma_0$ be some proper constants in $(0,1)$.

Let $(\Enc_0, \Dec_0)$ be an asymptotically good binary  $(n_0, k_0, d_0)$ error  correcting code with   $n_0 = O( \log \frac{1}{\eps})$, $k_0 = \gamma_0   n_0$, $d_0 =2\delta_0   n_0+1$. 


Encoding function $\Enc:\{0,1\}^{k = \Omega(n)} \rightarrow \{0,1\}^{n}$ is a random function  as follows.

\begin{enumerate}
\item On input $x \in \{0,1\}^k$, cut $x$ into blocks evenly of length $k_0 $ s.t.  $x = (x_1, \ldots, x_{k/  k_0 }) \in (\{0,1\}^{m_0})^{ k/k_0}$; (If there are less than $k$ bits in the last block then pad $0$'s)

\item  Compute $y'  =  \Enc_0(x_1)\circ \ldots \circ\Enc_0(x_{k/k_0}) $;

\item
Generate a  random permutation $\pi: [n ] \rightarrow [n ]$, where $n  = \kappa k/k_0$;

\item Let $y = ( y'_{\pi^{-1}(1)}, \ldots, y'_{\pi^{-1}(n)} ) $, i.e.   permute $y'$ using $\pi$;

\item Output $z = y \oplus w$ where $w$ is a uniform random string.
 
\end{enumerate}

Decoding function $\Dec: [k] \times \{0,1\}^{n } \rightarrow \{0,1\}^{k}$ is as follows.

\begin{enumerate}

\item  On the input $i$ and $ z $, XOR it with $w$ to get $y$; 

\item Use $r$ to reconstruct $\pi$;

\item Find $i'$ s.t. the $i$-th bit of the message is in the $i'$-th block of $ m_0 $ bits;

\item Query the bits $Q = \{ \, y_{\pi(j)} \mid j\in   ( \; (i'-1)m_0 , \, i'm_0 \;]  \, \} $ to get $y'_{i'}$;

\item $ x_{i'} = \Dec_0(y'_{i'}) $;

\item Output $x[i]$.
 
\end{enumerate}

\end{construction}

\begin{proof}[Proof of Theorem \ref{thm:LDCHamming1sr}]

Consider Construction \ref{construct:randEncLDCHammingsr}.

We claim that the random permutation $\pi$ and the $  z$ are independent. Actually for arbitrary permutation $\sigma$ and arbitrary $n$-bit string $a$,
$$  \Pr[ \pi = \sigma, z = a ]  = \frac{1}{n!} \times \frac{1}{2^n}. $$
Because  conditioned on any fixed permutation, to let $z=a$, there is a unique choice of $w$.

On the other hand, $\Pr[\pi = \sigma] = \frac{1}{n!}$ and $\Pr[z=a] = 2^{-n}$. This shows the independence of $\pi$ and $z$.

So even if the adversary operates based on knowing $z$, the set of positions of the tampered bits are independent of $\pi$.  

Since the number of errors is at most $\delta n$, the expected number of corrupted bits in $ y'_{i'} $ is at most $\delta n_0$. By Lemma \ref{lem:rponeset}, the probability that the number of corrupted bits in $y'_{i'}$ is at most $1.1 \delta n_0$, is at least $1- 2^{-0.01n_0/3} \geq 1- \eps/2$. By choosing $n_0$ to be large enough. Let $delta_0 = 1.1 \delta$. Then   $\Dec_0$ can recover the desired bit from   $y'_{i'}$ which has at most $1.1 \delta n_0$ errors.

Note that the construction runs in polynomial time because generating random permutation, permuting bits, $(\Enc_0, \Dec_0)$ are all in polynomial time. The number of queries is $n_0 = O(\log \frac{1}{\eps})$, since we only query bits in $y'_{i'}$. 
\end{proof}

 Next we give the construction for oblivious channels.

We use the almost $t$-wise independent random permutations instead of random permutations for this case. 
A random function $\pi: [n] \rightarrow [n]$ is an $\eps$
almost $t$-wise independent random permutation if for every $t$ elements $i_1, \ldots , i_t \in
 [n]$, $(\pi(i_1), \ldots, \pi(i_t))$ has statistical distance at most $\eps$ from  $(\pi'(i_1), \ldots, \pi'(i_t))$ where $\pi'$ is a random permutation over $[n]$.  Kaplan, Naor and Reingold \cite{kaplan2009derandomized} gives an explicit construction.

%
%
%
%
%
%
%
%
%
%

Next we show an efficient constructable $(n, k, \delta, q, \eps)$-locally decodable code, which is asymptotically good, with $ q = O(\log (1/\eps))$ and $\eps$ can be arbitrary in $[0, 1]$.

Before the main construction, we first show a binary ``locally decodable code'' which can recover the whole message locally when the message length is sufficiently small. 

\begin{lemma}
\label{lem:smallMsgLDC}
%

For every $k \leq cn $ with sufficiently small constant $c = c(\delta) < 1$, every sufficiently small constant $\delta$,  there is an explicit $(  n, k, \delta n   )$  binary ECC, which has a randomized decoding algorithm s.t. it can compute the message with success probability $1- 2^{-\Theta(k/\log n)}$, querying at most $q = O(k)$ bits.


\end{lemma}

\begin{proof}
Consider
the concatenation of an $(n_1, k_1, d_1 )$ Reed-Solomon code with  alphabet $\{0,1\}^{n_2 = O(\log n_1)}$, $n_1 = n/n_2$, $k_1 = k/k_2$, and an explicit $(n_2, k_2 = n_2 - 2d_2(\log \frac{n_2}{d_2} ), d_2 = \Theta(n)  )$ binary ECC.

The concatenated code is an  $(n_1n_2  , k_1k_2 , d_1d_2)$  code. Note that $n_1 n_2 = n$, $k_1 k_2 = k $. Also $d_1 = n_1 -k_1 +1 = \Theta(n_1)$,  $ d_2 = O(n_2)$ so $d_1 d_2 = O(n)$.  

We only need to show the decoding algorithm.

Given a codeword,
we call the encoded symbols (encoded by the second code) of the first code as blocks. 
The algorithm randomly picks $ 8k_1   $ blocks and query them. For each block, it calls the decoding function of the second code. After this we get $8k_1     $ symbols of the first code. Then we use the decoding algorithm of a $(8k_1, k_1, 7k_1+1)$ Reed Solomon code to get the message.

Next we argue that this algorithm successes with high probability.
Assume there are $\delta n = d_1d_2/8$ errors, then there are at most $d_1/4$ codewords of the second code which are corrupted for at least  $d_2/2$ bits. Since the distance of the second code is $d_2$,   there are $n_1 - d_1/4$ blocks can be decoded to their correct messages. 

Hence the expectation of the number of correctly recovered symbols of the first code is $ \frac{n_1 - d_1/4}{ n_1} 8k_1    \geq 6k_1$.
Thus by Chernoff Bound,  
with probability $1-2^{-\Theta(k_1)}$, there are at least $ 5k_1 $ symbols are correctly recovered. Note that if we look at the $8k_1$ queried blocks they also form a $(8k_1, k_1, 7k_1+1)$ Reed-Solomon code since it is the evaluation of the degree $k_1-1$ polynomial on $8k_1$ distinct values in the field $ \mathbb{F}_2^{n_2}$. Thus, our recovered symbols form a string which has only distance $3k_1$  from a codeword of the code. So by using the decoding algorithm of the code we can get the correct message.

\end{proof}

\begin{construction}
\label{construct:randEncLDCHamming1}
We construct an $(n, k = \Omega(n), \delta = \Theta(1), q= O(\log n \log \frac{1}{\eps}), \eps)$-locally decodable   binary code for arbitrary $\eps \in [0,1]$.
 
Let $\delta_0, \gamma_0$ be some proper constants in $(0,1)$.

Let $(\Enc_0, \Dec_0)$ be an asymptotically good binary  $(n_0, k_0, d_0)$ error  correcting code with   $n_0 = O( \log \frac{1}{\eps})$, $k_0 = \gamma_0   n_0$, $d_0 =2\delta_0   n_0+1$. 


Encoding function $\Enc:\{0,1\}^{k = \Omega(n)} \rightarrow \{0,1\}^{n}$ is a random function  as follows.

\begin{enumerate}
\item On input $x \in \{0,1\}^k$, cut $x$ into blocks evenly of length $k_0 $ s.t.  $x = (x_1, \ldots, x_{k/  k_0 }) \in (\{0,1\}^{m_0})^{ k/k_0}$; (If there are less than $k$ bits in the last block then pad $0$'s)

\item  Compute $y'  =  \Enc_0(x_1)\circ \ldots \circ\Enc_0(x_{k/k_0}) $;

\item
Generate a $\eps_{\pi} = \eps/10$-almost $\kappa = O(\log \frac{1}{\eps})$-wise independent random permutation $\pi: [n/2] \rightarrow [n/2]$ using Theorem \ref{thm:almosttwiserandperm}, where $n/2 = \kappa k/k_0$, the randomness used is $r\in \{0,1\}^{d_{\pi} = O(\kappa \log n + \log(1/\eps_{\pi}))}$;

\item Let $y = ( y'_{\pi^{-1}(1)}, \ldots, y'_{\pi^{-1}(n)} ) $, i.e.   permute $y'$ using $\pi$;

\item Use an   $( n/2,  |r| , \delta n )$ ECC from Lemma \ref{lem:smallMsgLDC} to encode $r$, getting $ z \in \{0,1\}^{ n/2} $;

\item Output $y\circ z$.
 
\end{enumerate}

Decoding function $\Dec: [k] \times \{0,1\}^{n } \rightarrow \{0,1\}^{k}$ is as follows.

\begin{enumerate}

\item  On the input $i$ and $y\circ z$, call the decoding algorithm from Lemma \ref{lem:smallMsgLDC} on $z$ to get $r$; 

\item Use $r$ to reconstruct $\pi$;

\item Find $i'$ s.t. $x[i]$ is in $x_{i'}$;

\item Query the bits $Q = \{ \, y_{\pi(j)} \mid j\in   ( \; (i'-1)m_0 , \, i'm_0 \;]  \, \} $ to get $y'_{i'}$;

\item $ x_{i'} = \Dec_0(y'_{i'}) $;

\item Output $x[i]$.
 
\end{enumerate}

\end{construction}

\begin{lemma}
\label{lem:constr3efficiency}
The encoding and decoding are both efficient.

\end{lemma}

\begin{proof}

Every step in the construction can be realized in polynomial time. The called functions in the construction are all from efficient constructions. 

\end{proof}

\begin{lemma}
\label{lem:constr3querynum}
The number of queries for the decoding algorithm is $q = O(\log n \log \frac{1}{\eps})$.

\end{lemma}

\begin{proof}

In step 1 of the decoding, calling the decoding algorithm from Lemma \ref{lem:smallMsgLDC} takes $O( d_{\pi})$ number of queries.

In step 4 of the decoding, the number of queries is $ n_0 $.

So the total number of queries is $n_0 + O(d_{\pi}) = O(\log n\log \frac{1}{\eps})$.

\end{proof}

\begin{lemma}
\label{lem:constr3Correct}
Assume there are $\delta n$ errors, then with probability $1-\eps$, the decoding can get $x$ correctly.

\end{lemma}

\begin{proof}

For $z$, the number of errors is at most $ \delta n$.
Thus by Lemma \ref{lem:smallMsgLDC}, the decoding can get $r$ correctly with probability $1-\eps/3 $  by letting the constant factor in $|r| = d_{\pi}$ to be large enough.

For $y$, the number of errors is at most $\delta n$. If $\pi$ is a uniform random permutation then
the expected number of  corrupted bits in $y'_{i'}$ is $ 2\delta n_0 $.

 By Lemma \ref{lem:rponeset}, the probability that the number of corrupted bits in $y'_{i'}$ is at most $ 2.2\delta n_0 $ is at least $1-2^{-0.01\delta n_0/3} \geq 1-\eps/3$ if letting the constant factor in $n_0$ to be large enough.

When $\pi$ is an $\eps_{\pi}$-almost $ n_0 $-wise independent permutation, the probability that the number of corrupted bits in $y'_{i'}$ is at most $ 2.2\delta n_0 $ is at least $1-2^{-0.01\delta n_0/3} +\eps_{\pi} \geq 1-\eps/2$.

Let $\delta = \delta_0/3$. Then the decoding $\Dec_0$ can recover the correct message from $2.2\delta n_0$ errors. 

Hence by a union bound, with probability at least $1-\eps/3 - \eps/2 \geq 1-\eps$, the decoding can recover $x_{i'}$ correctly. 

Note that once $x_{i'}$ is recovered, $x[i]$ is in it and thus is correctly recovered. 
\end{proof}

\begin{theorem}
\label{thm:LDCHamming1}
There is an efficient construction of $(n, k =\Omega(n) , \delta = \Theta(1), q = O(\log n\log \frac{1}{\eps}),  \eps)$ random LDC for oblivious channels, $\eps\in [0,1]$. 

\end{theorem}
\begin{proof}
We use Construction \ref{construct:randEncLDCHamming1}.
The theorem directly follows from Lemma \ref{lem:constr3efficiency}, \ref{lem:constr3querynum}, \ref{lem:constr3Correct}.

\end{proof}

\subsection{With Flexible Failure Probabilities}

Next we consider a stronger definition which can handle all decoding failure probabilities.

%
%
%
%
%
%
%
%

Intuitively we want to use the previous construction for different $eps$, and then combine them. It turns out this indeed can work but with some extra techniques while the information rate is $\Theta(1/\log {n})$.

\begin{construction}

\label{constr:flexOblivious}
We construct a $ (n, k = \Theta(n/\log n),  \delta = \Theta(1) ) $-locally decodable  binary code.

Let $(\Enc_0, \Dec_0)$ be a binary  $(n_0, k_0, d_0)$ error  correcting code with   $n_0 =  \gamma_0^{-1}  \log n $, $k_0 = \log n$, $d_0 =2\delta_0   n_0+1$, for constant $\gamma_0, \delta_0 \in [0,1]$. 

Let $(\Enc_i, \Dec_i)$ be the $(n_i = n/n_0, k, \delta_i = \Theta(1), q_i, \eps = 2^{-2^i})$ LDC from Theorem \ref{thm:LDCHamming1}, $i\in [\log n]$.

Encoding function $\Enc:\{0,1\}^{k = \Omega(n)} \rightarrow \{0,1\}^{n}$ is  as follows.

\begin{enumerate}
\item On input $x \in \{0,1\}^k$, compute $y_i = \Enc_i(x)$ for every $i\in [\log n]$.
 
\item  Let $M$ be a $ \log n \times n/\log n $ matrix s.t. $M[i][j] $ is the $j$-th bit of $y_i$.

\item Output $z = (\Enc_0(M_1), \ldots,  \Enc_0(M_{n_1})$ where $ M_j , j\in [n_1]$ is the $j$-th column of $M$.

\end{enumerate}

Decoding function $\Dec:[k]\times\{0,1\}^* \times [0,1] \rightarrow \{0,1\}$ is as follows.

\begin{enumerate}

\item On input $ u, z, \eps  $, find the smallest $i$ s.t. $2^{-2^i} \leq \eps$; If it cannot be found, then query the whole $z$; 

\item Compute $w = \Dec_i(u, y_i)$ but whenever $\Dec_i$ wants to query an $j$-th bit of $y_i$, we query $z_j$  and get $ y_i[j] $ from $\Dec_0(z_j)$;

\item Output $w$.

\end{enumerate}

\end{construction}

\begin{lemma}
\label{lem:constr4Efficiency}
The encoding and decoding are both efficient.

\end{lemma}

\begin{proof}

Every step in the construction can be realized in polynomial time. The called functions in the construction are all from efficient constructions. 

\end{proof}

\begin{lemma}
\label{lem:constr4QueryNum}
The number of queries for the decoding algorithm is $q = O(\log^2 n \log \frac{1}{\eps})$.

\end{lemma}

\begin{proof}

In step 2 of the decoding, calling the decoding algorithm from Theorem \ref{thm:LDCHamming1} takes $O(\log n \log \frac{1}{\eps})$ number of queries. For each such query, the algorithm actually queries 
$n_0$ bits. So the total number of queries is as stated.

\end{proof}

\begin{lemma}
\label{lem:constr4Correct}
Assume there are $\delta = \delta_0 \min_{j\in [n_1]}\{\delta_j\} $ fraction of errors, then with probability $1-\eps$, the decoding can get $x$ correctly.

\end{lemma}

\begin{proof}

As $\delta = \delta_0 \min_j\{\delta_j\} $, there are at most $ \delta_j n_1 $ number of $z_j$s which each has more than $\delta_0$ fraction of errors.

As a result, for $y_j$, there are  at most $ \delta_j n_1$ number of bits which cannot be computed correctly during decoding step 2. So $\Dec_j$ can compute $x[i]$ correctly by Theorem \ref{thm:LDCHamming1}.

\end{proof}

\begin{theorem}

There is an efficient construction of $(n, k =\Omega(n/\log n) , \delta = \Theta(1))$ random LDC which can recover any one message bit with probability $1-\eps$, for any $\eps \in [0,1]$, by doing $q = O(\log^2 n \log \frac{1}{\eps})$ queries. 

\end{theorem}

\begin{proof}

It directly follows from Construction \ref{constr:flexOblivious},  Lemma \ref{lem:constr4Efficiency}, \ref{lem:constr4QueryNum}, \ref{lem:constr4Correct}.

\end{proof}

Similarly we also have the following result for the case of shared randomness
\begin{theorem}

There is an efficient construction of $(n, k =\Omega(n/\log n) , \delta = \Theta(1))$ random LDC which can recover any one message bit with probability $1-\eps$, for any $\eps \in [0,1]$, by doing $q = O(\log^2 n \log \frac{1}{\eps})$ queries. 

\end{theorem}

\begin{proof}[Proof Sketch]

We modify Construction \ref{constr:flexOblivious} by letting   $(\Enc_i, \Dec_i), i\in [\log n]$ be from Theorem \ref{thm:LDCHamming1sr}.

Note that the outputted codeword is independent of the permutations used in $(\Enc_i, \Dec_i), i\in [\log n]$ due to the using of masks in the construction of Theorem \ref{thm:LDCHamming1sr}. This guarantees that a similar analysis can still work for this case.

\end{proof}

%% file: randEncLDCEdit.tex
\section{Constructions against Edit Error}
The constructions against edit error requires an additional layer of encoding. To make our proof easier to describe, we will utilize the following greedily constructed code from \cite{schulman1999asymptotically}. 

\begin{lemma}
	[Implicit from \cite{schulman1999asymptotically}]
	\label{greedily constructed insdel code}
	For small enough $\delta$, there exists code $\mathbb{C}: \{0,1\}^n \rightarrow \{0,1\}^{4n}$ such that the edit distance between any two codewords is at least $\delta$ and for any codeword, every interval has at least half of 1's.
\end{lemma}
The proof uses the same argument from \cite{schulman1999asymptotically}. The only difference is that their proof only considers insertion and deletion. Our slightly modified proof includes substitution as a type of error.
\begin{proof}
	
	\par For a codeword $y$, the number of words in $\{0,1\}^n$ that is within edit distance $2d$ is at most $\binom{n}{2d}^22^{2d}$. That is because each substitution can be replaced by a deletion and an insertion. For any two codewords $x$ and $y$ such that $d_E(x,y)\leq d$, $x$ can be transformed into $y$ with at most $2d$ insertions and $2d$ deletions. There are $\binom{n}{2d}$ ways to delete $2d$ characters from $y$, $\binom{n}{n-2d}$ ways to put the remaining characters in proper position of $x$, and $2^{2d}$ ways to choose to the inserted characters. 
	\par Substitute $d$ with $\delta n$ for some small constant $\delta$, since $\binom{n}{k}\leq (\frac{ne}{k})^k$, we have 
	\[\binom{n}{2\delta n}^2 2^{2\delta n}\leq (\frac{ne}{2\delta n})^{4\delta n}2^{2\delta n} \leq (\frac{e^4}{4\delta^4})^{\delta n} = (2^{\log \frac{e^4}{4\delta^4}} )^{\delta n}\]
	Pick $\delta\in (0,0.018)$, we have $\binom{n}{2\delta n}^2 2^{2\delta n}\leq 2^{1/2 n}$. This gives a greedy construction of edit error code with rate $\frac{1}{2}$ and tolerates a $\delta$ fraction of edit error for constant $\delta\leq 0.018$. We then ensure every interval has at least half 1's by inserting 1's into every other position. Note that this operation will not decrease the edit distance. It gives us a code with rate $1/4$.
\end{proof}
\par  In the discussion below, we will use the following error correcting code for edit distance: 
\[\mathbb{C}_0:\{0,1\}^N\rightarrow \{0,1\}^{5N}\]
such that for each message $S\in \{0,1\}^N$, $\mathbb{C}_0(S)$ is composed of two parts. The first part is a codeword $\mathbb{C}'_0(S)$ where $\mathbb{C}'_0:\{0,1\}^N\rightarrow \{0,1\}^{4N}$ is a greedily constructed code that can tolerate a constant fraction of edit errors. The second part is a buffer of 0's of length $N$. That is, $\mathbb{C}_0(S) = \mathbb{C}'_0(S)\circ 0^N$. We assume the normalized edit distance between any two codewords of $\mathbb{C}_0(S)$ is at least $\delta'$.

Another key component in our construction against edit error is a searching algorithm developed by \cite{ostrovsky2015locally}. 

Their work provides an algorithm for searching from a weighted sorted list $L$ with a constant fraction of errors. An element $(i,a_i)$ in the sorted list is composed of two parts, an index $i$ and content in that element, $a_i$. All elements in the list are sorted by their index, i.e. the $j$-th element in $L$ is $(j,a_j)$. Beyond that, each element is equipped with a non-negative weight. When sampling from the list, each element is sampled with probability proportional to its weight. 


In \cite{ostrovsky2015locally}, the authors proved the following result.

\begin{lemma}
	\label{search}[Theorem 16 of \cite{ostrovsky2015locally}]
	Assume $L'$ is a corrupted version of a weighted sorted list $L$ with $k$ elements, such that the total weight fraction of corrupted elements is some constant $\delta$ of the total weight of $L$. And the weights have the property that all sequences for $r\geq3$ elements in the list have total weight in the range $[r/2,2r]$. Then, there is an algorithm for searching $L$. For at least a $1-O(\delta)$ fraction of the original list's elements, it recovers them with probability at least $1-\mathsf{neg}(k)$. It makes a total of $O(\log^3 k)$ queries.
\end{lemma}
This lemma enables us to search from a weighted sorted list (with corruption) with few ($\polylog k$) queries. To make our proof self-contained, we will describe how to turn the encoded message in our construction into a weighted sorted list in the proof.

We are now ready to describe our construction.

\subsection{Construction with Fixed Failure Probability}

\subsubsection{Shared Randomness}

We first give the construction of an $(n,k = \Omega(n),\delta = O(1),q=\polylog k\log\frac{1}{\eps}, \eps)$ randomized LDC. As before, We start with the construction assuming shared randomness. 

\begin{construction}
	\label{construct:randEncLDCEdit}
	\par We construct an $(n,k=\Omega(n),\delta = O(1),q =\polylog k\log\frac{1}{\eps},\eps)$-LDC with randomized encoding.
	
	Let $\delta_0$, $\gamma_0$ be some proper constants in $(0,1)$.
	
	Let $(\Enc_0,\Dec_0)$ be an asymptotically good $(n_0,k_0,d_0)$ error correcting code for Hamming distance over alphabet $\{0,1\}^{10\log k}$. Here we pick $n_0 = O(\log\frac{1}{\eps})$, $k_0 = \gamma_0n_0$, and $d = 2\delta_0 n_0 + 1$.
	
	Let $\pi$ be a random permutation. And $r_i\in \{0,1\}^{\log k}$ for $i\in[\frac{n_0 k}{k_0 \log k}]$ be $\frac{n_0 k}{k_0 \log k}$ random masks. Both $\pi$ and $r_i$'s are shared between the encoder and decoder.
	
	Let $\mathbb{C}_0:\{0,1\}^{2\log k}\rightarrow \{0,1\}^{10\log k}$ be the asymptotically good code for edit error described previously that can tolerate a $\delta'$ fraction of edit error.
	
	The encoding  function $\Enc:\{0,1\}^k\rightarrow \{0,1\}^n$ is a random function as follows
	\begin{enumerate}
		\item On input $x\in \{0,1\}^k$, view $x$ as a string over alphabet$\{0,1\}^{\log k}$of length $k/\log k$. We write $x = x_1x_2\cdots x_{k/\log k}\in (\{0,1\}^{\log k})^{k /\log k}$;
		\item Divide $x$ into small blocks of length $k_0$, s.t. $x=B^{(0)}_1\circ B^{(0)}_2 \circ  \cdots \circ B^{(0)}_{k/(k_0\log k)})$. Here, $B^{(0)}_i\in (\{0,1\}^{\log k})^{k_0}$ for $i\in [k/(k_0\log k)]$ is a concatenation of $k_0$ symbols in $x$;
		\item Encode each block with $\Enc_0$. Concatenate them to get $y^{(1)} = \Enc_0(B^{(0)}_1)\circ \Enc_0(B^{(0)}_2)\circ \cdots \circ \Enc_0(B^{(0)}_{k/(k_0\log k)})$. Notice that each $\Enc_0(B^{(0)}_1)$ is composed of $n_0$ symbols in $\{0,1\}^{\log k}$. Write $y^{(1)}$ as a string over alphabet $\{0,1\}^{\log k} $, we have $ y^{(1)} = B^{(1)}_1\circ B^{(1)}_2\circ \cdots \circ B^{(1)}_{\frac{n_0k}{k_0\log k}}$ such that  $B^{(1)}_i\in \{0,1\}^{\log k} $;
		\item Let $N = \frac{n_0k}{k_0\log k}$. Permute these $N$ symbols of $y^{(1)}$ with permutation $\pi$ to get $y^{(2)}=B^{(2)}_1\circ B^{(2)}_2\circ \cdots\circ B^{(2)}_{N}$ such that $B^{(2)}_{\pi(i)} = B^{(1)}_{i}$;
		\item Let $b_i\in \{0,1\}^{\log k}$ be the binary representation of $i \in [N]$. This is fine since $N<k$ when $k$ is larger enough. We compute $B^{(3)}_i = \mathbb{C}_0(b_i\circ (B^{(2)}_i\oplus r_i))\in \{0,1\}^{10\log k}$ for each $i\in [T]$. We get $y = B^{(3)}_1\circ  B^{(3)}_2\circ \cdots \circ B^{(3)}_{N}\in (\{0,1\}^{10\log k })^{N}$;
		\item Output $y$.
	\end{enumerate}
	The decoding function $\Dec: [k]\times \{0,1\}^n\rightarrow \{0,1\}$ takes two inputs, an index $i_0\in [k]$ of the message bit the decoder wants to know and $\omega \in \{0,1\}^n $, the received (possibly corrupted) codeword. It proceeds as follows 
	\begin{enumerate}
		\item On input index $i_0$ and the received codeword $\omega$. We assume the $i_0$-th bit lies in $B^{(0)}_i$, i.e. the $i$-th block of $x$;
		\item Notice that $\Enc_0(B^{(0)}_i) = B^{(1)}_{(i-1)n_0+1}\circ B^{(1)}_{(i-1)n_0+2}\circ \cdots \circ B^{(1)}_{in_0}$. For each $j\in  \{(i-1)n_0+1, (i-1)n_0+2,\ldots, i n_0\}$, search from $y$ to find the block $B^{(3)}_{\pi(j)}$ using the algorithm from \cite{ostrovsky2015locally}. Then we can get a possibly corrupted version of $\Enc_0(B^{(0)}_i)$
		\item Run the decoding algorithm $\Dec_0$ to find out $B^{(0)}_i$. This gives us the $i_0$-th bit of $x$.
	\end{enumerate}
\end{construction}

\begin{lemma}
	\label{fixed_sharedrand}
	The above construction \ref{construct:randEncLDCEdit} gives an efficient $(n,k = \Omega(n),\delta = O(1),q=\polylog k\log\frac{1}{\eps}, \eps)$ randomized LDC against edit error.
\end{lemma}

\begin{proof}
	We first show both the encoding and decoding can be done in polynomial time. Although for the ease of description, we picked the greedily constructed code $\mathbb{C}_0$, which can be inefficient. The codeword size of $\mathbb{C}_0$ is $O(\log k)$. Decoding one block can be finished in time polynomial in $k$. And for our purpose, we need to encode $O(k/\log k)$ blocks and decode $\polylog n\log {\frac{1}{\eps}}$ blocks. Thus, the additional time caused by this layer of encoding is polynomial in $n$. For the rest part, the analysis is similar to that of the Hamming case. Thus, our code is polynomial time. We note that $\mathbb{C}_0$ can be replaced by an efficient code for edit error.

	We use the same notation as in our construction. Let $y = B^{(3)}_1\circ B^{(3)}_2\circ \cdots \circ B^{(3)}_N$ be the correct codeword. We denote the length of each block $B^{(3)}_j$ by $b = 10\log k$ and if we view $y$ as a binary string, the length of $y$ is $N'( = Nb)$, which is $O(k)$.  We call the received (possibly corrupted) codeword $\omega$. Since we can always truncate (or pad with 0) to make the length of $\omega$ be $N'$, which will only increase the edit distance by a factor no more than 2. Without loss of generality, we assume the length of $\omega$ is also $N'$.
	
	The decoding function $\Dec$ has two inputs: an index $i_0\in [k]$ of the message bit the decoder wants to see and the received codeword $\omega$. The decoder also has access to the shared randomness, which has two parts: a permutation $\pi$ and $N$ random masks, each of length $\log k$. The first step is to figure out the indices of the blocks it wants to query by using the shared permutation $\pi$. Then, we can query these blocks one by one. We assume the $i_0$-th bit lies in $B^{(0)}_i$. Then, notice that $\Enc_0(B^{(0)}_i) = B^{(1)}_{(i-1)n_0+1}\circ B^{(1)}_{(i-1)n_0+2}\circ \cdots \circ B^{(1)}_{in_0}$. Our goal is to find the block $B^{(3)}_{\pi(j)}$ for each $j\in  \{(i-1)n_0+1, (i-1)n_0+2,\ldots, i n_0\}$.
	
	One thing we need to clarify is how to query a block since we do not know the starting point of each block in the corrupted codeword. This is where we use the techniques developed by \cite{ostrovsky2015locally}. In the following, assume we want to find the block $B^{(3)}_{i}$.
	
	We can view $y$ as a weighted sorted list $L$ of length $N$ such that the $i$-th element in $L$ is simply $B^{(3)}_i$ with weight $b$. We show that $\omega$ can be viewed as a corrupted list $L'$. There is a match from $y$ to $\omega$ which can be described by a function $f:[N]\rightarrow [N]\cup \{\perp\}$. If the $i$-th bit is preserved after the edit error, then $f(i) = j$ where $j$ is the position of that particular bit in $\omega$. If the $i$-th bit is deleted, then $f(i) = \perp$. 
	
	We say the $i$-th bit in $y$ is preserved after the corruption if $f(i)\neq \perp$. For each block  $B^{(3)}_{i} $ that is not completely deleted, let $v_{i} = f(u_i)$, such that $u_{i}$ is the index of the first preserved bit in the block $B^{(3)}_{i} $. We say the block $B^{(3)}_{i} $ is \emph{recoverable} if $\Delta_E(B^{(3)}_{i}, \omega_{[v'_i,v'_i+b-1]})\leq \delta'$ for some index $v'_i$ that is at most $\delta'b$ away from $v_i$.
	
	\begin{claim}
		\label{claim}
		If $\Delta_E(\omega,y)\leq \delta$, each block is recoverable with probability at least $1-\delta/\delta'$.
	\end{claim}
	\begin{proof}
		\par The code $\mathbb{C}_0$ is greedily constructed and resilient to $\delta'$ fraction of error. Here, $\delta$ is picked smaller than $\delta'$. To make an block corrupted, the adversary needs to produce at least a $\delta'$ fraction of edit error in that block. The adversary channel can corrupt at most $\frac{\delta}{\delta'}$ fraction of all blocks.
	\end{proof}

	The searching algorithm requires sampling some elements from the sorted list $L'$ corresponding to the corrupted codeword $\omega$. We now explain how to  do the sampling. We start by first randomly sample a position $r$ and read a substring of $2r+1$ bits $\omega_{[r-b,r+b]}$ from $\omega$. Then, we try each substring in $\omega_{[r-b,r+b]}$ of length $b$ from left to right until we find first substring that is $\delta_0$-close to some codeword of $\mathbb{C}$ under the normalized edit distance. If we did find such a substring, we decode it and get $b_j\circ B^{(2)}_{j} \oplus r_j$. Otherwise, output a special symbol $\perp$.
	
	We will regard consecutive intervals in $\omega$ as elements in the weighted sorted list where weight of the element is simply the length of the corresponding interval. The correspondence can be described as following. 
	
	For a recoverable block $B^{(3)}_i$, again, we let $v_{i} = f(u_i)$, such that $u_{i}$ is the index of the first preserved character in block $B^{(3)}_{i} $. We want to find an interval $I_i$ in $\omega$ to represent $B^{(3)}_{i} $ in the list $L'$. Since $B^{(3)}_i$ is a codeword of $\mathbb{C}_0$ of length $b$. Every interval in its first $4/5$ part has at least half of 1's and the last $1/5$ part are all 0's. For any $v\in[v_i-b+2\delta' b,v_i-2\delta'b]$, we know $\omega_{[v,v+b-1]}$ can not be $\delta'$ close to any codeword in $\mathbb{C}_0$. It is because the last $1/5$ part of $\omega_{[v,v+b-1]}$ contains at least $\delta' b$ 1's. Thus, in the sampling procedure, if $r\in[v_i+2\delta' b,v_i+b]$, it will return $B^{(3)}_{i} $. The length of $I_i$ is at least $1-2\delta' b$. Let $I_i$ be the maximal inteval containing $[v_i+2\delta' b,v_i+b]$, such that, if sampling $r$ in $I_i$, it will output same codeword in $\mathbb{C}$. Also, we argue the length of the inteval $I_i$ is no larger than $1+ 4 \delta'$. Since if $r \geq v_i+(1+2\delta') b$ or $r \leq v_i-2\delta' b$, any substring of length $b$ in $[r-b,r+b] $ is at least be $\delta'$ far from $B^{(3)}_{i} $ and thus output a different codeword. We note due to the existence of adversary insertion, it is possible to get $B^{(3)}_{i} $ outside of $I_i$. But this does not affect our analysis below.
	
	The above method for sampling a block from $\omega$ gives us a natural way to interpret $\omega$ as a corrupted weighted list $L'$. The length of $\omega$ is equal to the sum of weights of all its elements. For each recoverable block, the interval $I_i$ as described above is an element with weight equal to its length. We note that all $I_i$'s are disjoint. We consider the remaining characters in $\omega$ as corrupted elements in the list $L$. For those elements, the sampling algorithm will output wrong result or $\perp$. Since the interval of the remaining elements can be large, we divide such interval into small intervals each has length no larger than $b$ and consider each small intervals as an element. 
	
	Next, we argue the weight fraction of corrupted elements is small. Due to claim \ref{claim}, there are at least a $1-\delta/\delta'$ fraction of blocks recoverable after the corruption. For each recoverable block $B^{(3)}_{i}$, we can find a interval $I_i$ in $\omega$ with length at least $(1-2\delta_0)b$. Thus, the total weight of uncorrupted elements from the original list $L$ is at least $(1-\delta/\delta')(1-2\delta')N = (1-2\delta-2\delta_0-\delta/\delta_0)N$. Let $\delta_1 = 2\delta+2\delta'+\delta/\delta' = O(\delta)$, we know the total weight of corrupted elements is at most $\delta_1$ fraction of the total weight of list $L'$. We assume $\delta_1$ is small by properly picking $\delta$ and $ \delta'$.
	
	By lemma \ref{search}, $1-\delta_2$ fraction of elements can be decoded correctly with high probability ($1-\mathsf{neg}(n)$) for some constant $\delta_2 = O(\delta)$. Since the content of each block is protected by random masks. The adversary can learn nothing about the random permutation $\pi$ used for encoding. Each block is recoverable with same probability. We want to search from $y$ to find all blocks $B^{(3)}_{\pi(j)}$ for each $j\in  \{(i-1)n_0+1, (i-1)n_0+2,\ldots, i n_0\}$ to get a possibly corrupted version of  $\Enc_0(B^{(0)}_i)$. The proof of success probability follows the same concentration bound as in the Hamming case. Here, the random masks serve the same purpose as in the Hamming case. The content of each block (not the index) is not known to the adversary. Thus, the adversary can learn nothing about the permutation or the message. By picking properly the code $(\Enc_0,\Dec_0)$, using the concentration bound, the fraction of unrecoverable blocks among these $n_0$ blocks is larger than $1-1.1\delta_2$ with probability at least $1-\eps$.
	
	Finally, we count the number of queries made. Searching one block queries $\polylog k$ symbols from the corrupted codeword $\omega$. We need to search $n_0 = O(\log \frac{1}{\eps})$ blocks. The total number of queries made is $\polylog k \log \frac{1}{\eps}$. 
\end{proof}

\subsubsection{Oblivious Channel}
Now, we consider the oblivious channel. In this model, we assume the adversary can not read the codeword. Same as the Hamming case, we want to send a description of a permutation to the decoder. The decoder can then use description to recover the permutation $\pi$ which is used to encode the message. We can use the same binary error correcting code from Lemma \ref{lem:smallMsgLDC} to encode the description.


We now give the construction.

\begin{construction}
	\label{construct:randEncLDCEditOblivious}
	\par We construct an $(n,k=\Omega(n),\delta = O(1),q = \polylog k\log\frac{1}{\eps},\eps)$-LDC with randomized encoding against the oblivious channel model.
	
	Let $\delta_0$, $\gamma_0$ be some proper constants in $(0,1)$.
	
	Let $(\Enc_0,\Dec_0)$ be an asymptotically good $(n_0,k_0,d_0)$ error correcting code for Hamming distance on alphabet set $\Sigma = \{0,1\}^{10\log k}$. Here we pick $n_0 = O(\log\frac{1}{\eps})$, $k_0 = \gamma_0n_0$, and $d = 2\delta_0 n_0 + 1$.

	Let $\mathbb{C}_0:\{0,1\}^{2\log k}\rightarrow \{0,1\}^{10\log k}$ be the asymptotically good code for edit error as described above that can tolerate a $\delta'$ fraction of edit error.
	
	The encoding  function $\Enc:\{0,1\}^k\rightarrow \{0,1\}^n$ is a random function as follows:
	\begin{enumerate}
		\item Encode the message $x$ with Construction \ref{construct:randEncLDCEdit} without doing steps 4,5, and 6. This does not require knowing permutation $\pi$. View the sequence we get as a binary string. View the output as a binary string $y$ with length $n/10$;
		\item Let $N_1 = \frac{n}{20\log k}$. Generate a random seed $r\in \{0,1\}^{d}$ of lengh $d$. Here, $d = O(\kappa \log N_1 + \log(1/\eps_{\pi})) = O(\log n\log \frac{1}{\eps})$. Use $r$ to sample a $\eps_{\pi} = \eps/10$-almost $\kappa = O(\log \frac{1}{\eps})$-wise independent random permutation $\pi: [N_1] \rightarrow [N_1]$ using Theorem \ref{thm:almosttwiserandperm};
		\item Let $\delta_1$ be a properly chosen constant larger than $\delta$. Encoding $r$ with an $( n/20,  d , \delta_1 n )$ error correcting code from Lemma \ref{lem:smallMsgLDC}, we get $ z \in \{0,1\}^{ n/10} $;
		\item View $y\circ z$ as a sequence in $\{0,1\}^{n/10}$. Divide $y \circ z\in \{0,1\}^{n/10}$ into small blocks of size $\log k$. Write $y\circ z$ as $y\circ z = B_1\circ B_2\circ \cdots \circ B_{n/(10\log k)}$;
		\item Permute first half of blocks with random permutation $\pi$ to get $u' = B^{(1)}_1\circ B^{(1)}_2\circ \cdots \circ B^{(1)}_{n/(10\log k)}$ such that $B^{(1)}_{\pi(i)} = B_i$ for $i\leq n/(20 \log k)$ and $B^{(1)}_i = B_i$ for $i$ larger than $n/(20 \log k)$;
		\item Let $b_i$ be the binary representation of $i$, for each $i\in n/(10 \log k)$, encode $b_i\circ B^{(1)}_i$ with code $\mathbb{C}_0$ to get $B^{(2)}_i$;
		\item Output $u = B^{(2)}_1\circ B^{(2)}_2\circ \cdots \circ B^{(2)}_{n/(10\log k)}$.
		
	\end{enumerate}

	The decoding function $\Dec: [k]\times \{0,1\}^n\rightarrow \{0,1\}$ takes two inputs, an index $i_0\in [k]$ of the message bit the decoder wants to know and $\omega \in \{0,1\}^n $, the received (possibly corrupted) codeword. It proceeds as follows:
	
	\begin{enumerate}
		\item Search for at most $O(\log n \log \frac{1}{\eps})$ blocks to decode the random seed $r$;
		\item Use $r$ to generate the random permutation $\pi$;
		\item Run the same decoding algorithm as in the Construction \ref{construct:randEncLDCEdit} on the first half of $\omega$ to decode $x_{i_0}$.
	\end{enumerate}
	
\end{construction}

\begin{lemma}
	\label{fixed_oblivious}
	The above construction \ref{construct:randEncLDCEdit} gives an efficient $(n,k = \Omega(n),\delta = O(1),q=\polylog(n)\log\frac{1}{\eps}, \eps)$ randomized LDC against an oblivious channel with edit error.
\end{lemma}

\begin{proof}
	The proof of efficiency of this construction follows directly from the proof of Lemma \ref{fixed_sharedrand}. 
	
	We denote the uncorrupted codeword by $u$ and the received codeword by $\omega$. We assume $\Delta_E(u,\omega)\leq \delta$. We first divide it into two parts $\omega_1$ and $\omega_2$ each of equal length such that $\omega = \omega_1\circ \omega_2$. Similarly, we have $u = u_1\circ u_2$ with $|u_1| = |u_2| = n/2$. We have $\Delta_E(u_1,\omega_1)\leq \delta$ and $\Delta_E(u_2,\omega_2)\leq \delta$. 
	
	The first step is to decode $r$. By Lemma \ref{lem:smallMsgLDC}, we need to query $O(d) $ bits of $z$ to decode $r$ with high probability $1-\mathsf{neg}(k)$. Thus, we need to query $O(\log k \log \frac{1}{\eps})$ blocks in the second half of $\omega$. We can then use the same searching algorithm described in the proof of Lemma \ref{fixed_sharedrand} to query each block. Since we randomly choose blocks to query, we can recover $r$ with probability $1-\mathsf{neg}(k)$ by making $\polylog k \log \frac{1}{\eps}$ queries to $\omega_2$.
	
	Once we know $r$, we can use it to generate the random permutation $\pi$. We run the decoding algorithm from Construction \ref{construct:randEncLDCEdit} with input $i_0$ and $\omega_1$. The remaining part is the same as the proof of Lemma \ref{fixed_sharedrand}. Notice that recovering $r$ takes $\polylog k \log \frac{1}{\eps}$ queries. The query complexity for our code is still $\polylog k \log \frac{1}{\eps}$.
\end{proof}

\subsection{Construction with Flexible Failure Probability}

\subsubsection{Shared Randomness}
We now give the construction for flexible failure peobability against edit error. 

\begin{construction}
	\label{construct:randEncLDCEditFlex}
	We construct an $(n,k = \Theta(n/\log n),\delta = O(1))$ randomized LDC with query complexity function $q = \polylog k \log \frac{1}{\eps}$ for any given failure probability $\eps$.
	
	Let $\pi$ be a random permutation. And $r_i\in \{0,1\}^{\log k}$ for each $i\in {n/(10\log k)}$ be random masks. Both $\pi$ and $r_i$'s are shared between the encoder and decoder.
	
	With properly picked constant $\gamma_0$ and $\delta_0$. Let $(\Enc_0, \Dec_0)$ be a $(n_0, k_0, d_0)$ error  correcting code on alphabet set $\{0,1\}^{\log k}$. Here, $n_0 =  \gamma_0^{-1}  \log n $, $k_0 = \log n$, $d_0 =2\delta_0   n_0+1$, for constant $\gamma_0, \delta_0 \in [0,1]$.

	For each $i\in [\log n]$, let $\Enc_i$ be the encoding algorithm of an $(n_i, k, \delta_i = O(1), q_i, \eps = 2^{-2^i})$ randomized LDC from Construction \ref{construct:randEncLDCEdit} without steps 4,5, and 6. Let $\Dec_i$ be the correponding decoding algotithm.
	
	Let $\mathbb{C}_1:\{0,1\}^{3\log k}\rightarrow \{0,1\}^{15\log k}$ be the asymptotically good code for edit error as described above that can tolerate a $\delta'$ fraction of edit error.

	Encoding function $\Enc:\{0,1\}^{k = \Omega(n)} \rightarrow \{0,1\}^{n}$ is  as follows.
	
	\begin{enumerate}
		\item On input $x \in \{0,1\}^k$, compute $y_i = \Enc_i(x)$ for every $i\in [\log n]$. As before, we view $y_i = B^i_1\circ  B^i_2\circ \cdots \circ B^i_{N}$ as the concatenation of $N$ symbols over alphabet $\{0,1\}^{\log k}$;
		
		\item  Let $M$ be a $ \log n \times N$ matrix such that $M[i][j]\in \{0,1\}^{\log k}$ is the $j$-th symbol of $y_i$;
		
		\item Let $M_j$ be the $j$-th column of $M$ for each $j\in [N]$. We compute $z_j  = \Enc_0(M_j)$. Notice that $z_j = \Enc_0(M_j)\in (\{0,1\}^{\log k})^{n_0}$ is a string of length $n_0$ over alphabet $\{0,1\}^{\log k}$;
		
		\item Let $y^{(0)}$ be the concatenation of $z_j$'s for $j\in [N]$. Then $y^{(0)} = z_1\circ z_2 \cdots \circ  z_N$ is a string of length $ n_0N$ over alphabet $\{0,1\}^{\log k}$. Let $n = 15\log (k)n_0N = O(k\log k)$;
		
		
		\item Permute the symbols of $y^{(0)}$ with permutation $\pi$ to get $y^{(1)} = B^{(1)}_1\circ  B^{(1)}_2\circ \cdots \circ B^{(1)}_{n_0N}$ such that $B^{(1)}_{\pi(i)} = B^{(0)}_i$;
		
		\item Since $n_0N = O(k)$, we assume $n_0N = ck$ for some constant $c$. When $k$ is large enough, we assume $ck<k^2$. Let $b_i\in \{0,1\}^{2\log k}$ be the binary representation of $i$ for each $i\in [n_0N]$. We compute $B^{(2)}_i = \mathbb{C}_1(b_i\circ (B^{(1)}_i\oplus r_i))\in \{0,1\}^{10\log k}$ for each $i\in [n_0N]$. We get $y = B^{(2)}_1\circ  B^{(2)}_2\circ \cdots \circ B^{(2)}_{N'})\in (\{0,1\}^{15\log k })^{n_0N}$;
		
		\item Output $y$ as a binary string of length $n$.

	\end{enumerate}
	
	Decoding function $\Dec$ is a randomized algorithm takes three inputs: an index of the bit the decoder wants to see,the received codeword $\omega$, and the desired failure probability $\eps$. It can be described as follows.
	
	\begin{enumerate}

		\item On input $ i_0, \omega, \eps  $, find the smallest $i$ such that $2^{-2^i} \leq \eps$. If it cannot be found, then query the whole $\omega$; 
		
		\item Compute $w = \Dec_i(i_0, y_i)$ but whenever $\Dec_i$ wants to query an $j$-th symbol of $y_i$, we decode $z_j$. Notice that $z_j = \Enc_0(M_j)$ contains $n_0$ symbols over alphabet $\{0,1\}^{\log k}$. Each symbol is encoded in a block in $y$. The indices of these blocks can be recovered by $\pi$. We search each of these $n_0$ blocks in $\omega$. We can then get the $j$-th block of $y_i$ from $\Dec_0(z_j)$;
		
		\item Output $w$.
		
	\end{enumerate}

\end{construction}

\begin{lemma}
	\label{flex_sharerand}
	The above construcion \ref{construct:randEncLDCEditFlex} gives an efficient $(n,k = \Theta(n/\log n),\delta = O(1))$ randomized LDC that can recover any bit with probability at least $1-\eps$ for any constant $\eps \in (0,1)$, with query complexity $q = \polylog k \log \frac{1}{\eps}$.
\end{lemma}
\begin{proof}
	In the encoding, we need to encode $n_0N = O(n/\log n)$ blocks with edit error code $\mathbb{C}_0$. And in the decoding, we need decode at most $\polylog n \log \frac{1}{\eps}$ blocks. The time caused by the second layer is polynomial in $n$. Thus, our construction is efficient.
	
	After a $\delta$ fraction of edit error. There is some constant $\delta_1 = O(\delta)$ such tath for at least $1-\delta_1$ fraction of blocks, we can use the searching algorithm to find and decode them correctly with probability $1-\mathsf{neg}(n)$. Thus, as is described in the construction, when ever we want to query the $j$-th symbol of $y_i$, we fist decode $z_j$. The symbols of $z_j$ is encoded in $n_0$ blocks in $\omega$. With $\pi$, we know the indices of these $n_0$ blocks. We can than perform the search as described in the proof of Lemma \ref{fixed_sharedrand}. By picking proper $n_0$, the decoding of $z_j$ fails with a small probability. The rest of analysis follows directly from the proof of Lemma \ref{fixed_sharedrand}.

	For the query complexity, we need to query $O(n_0\log \frac{1}{\eps})$ blocks. Since $n_0 = O(\log n) = O(\log k)$, the total number of queries made is still $\polylog k \log \frac{1}{\eps}$.
\end{proof}

\subsubsection{Oblivious Channel}
Our construction for flexible failure probability against oblivious channel combines Construction \ref{construct:randEncLDCEditOblivious} and Construction \ref{construct:randEncLDCEditFlex}. More specifically, following the Construction \ref{construct:randEncLDCEditFlex}, we replace the encoding functions $\Enc_i$ for each $i\in [\log n]$ from Construction \ref{construct:randEncLDCEdit} with the encoding functions from Constructon \ref{construct:randEncLDCEditOblivious}. The analysis follows directly. We omit the details.